\documentclass[runningheads, 11pt]{llncs}

\usepackage{amsmath}
\usepackage{enumitem}
\usepackage{amssymb}
\usepackage[hidelinks]{hyperref}
\usepackage{algorithm}
\usepackage{algorithmic}
\DeclareMathAlphabet\rsfscr{U}{rsfso}{m}{n}

\def \NP   	{\mathbf{NP}}

\def \DSPACE 	{\mathbf{DSPACE}}
\def \IP        {\mathbf{IP}}

\def \PSPACE    {\mathbf{PSPACE}}

\def \VDF       {\mathbf{VDF}}
\def \TM    	{\mathsf{TM}}
\def \PTM    	{\mathsf{PTM}}
\def \O     	{\mathcal{O}}
\def \Z 	{\mathbb{Z}}

\def \F 	{\mathbb{F}}

\def \L 	{\mathcal{L}}

\def \P     	{\mathsf{P}}

\def \HALT   	{\textsc{SPACEHALT}}
\def \TQBF   	{\textsc{TQBF}}

\def \RO	{\mathsf{H}}
\def \adv       {\mathcal{A}}
\def \adf       {\mathcal{B}}

\def \pp   	{pp}
\def \setup	{\mathsf{Setup}}
\def \eval	{\mathsf{Eval}}
\def \vdf	{\textsc{VDF}}
\def \verify	{\mathsf{Verify}}

\def \ips	{(\prv \leftrightarrow  \vrf)}
\def \vdfs	{\langle \prv \leftrightarrow  \vrf \rangle}
\def \advs	{\langle \adv \leftrightarrow  \vrf \rangle}
\def \adfs	{\langle \adf \leftrightarrow  \vrf \rangle}
\def \negl	{\mathtt{negl}}
\def \poly	{\mathtt{poly}}
\def \st	{\mathtt{state}}
\def \mod       {\; \mathbf{mod} \;}
\def \multgroup#1{(\mathbb{Z}/#1\mathbb{Z})^\times}
\def \vrf       {\mathcal{V}}
\def \prv       {\mathcal{P}}
\def \X         {\mathcal{X}}
\def \Y         {\mathcal{Y}}

\begin{document}
\title{$\VDF\subsetneq\PSPACE$}
%

%
\author{Souvik Sur}

\institute{
\email{souviksur@gmail.com}
}
\titlerunning{$\VDF\subsetneq\PSPACE$}
\maketitle              

\begin{abstract}
Verifiable delay functions ($\vdf)$ are functions that take
a specified number of sequential steps to be evaluated but can be verified efficiently. 
In this paper, we show that every $\vdf$ is provable in $\PSPACE$ but
every language in $\PSPACE$ does \emph{not} admit to a $\vdf$.

\end{abstract}

\section{Introduction}\label{introduction}
 In 1992, Dwork and Naor introduced the very first notion 
 of $\vdf$ under a different nomenclature ``pricing function" \cite{Dwork1992Price}.
 It is a computationally hard puzzle that needs to be solved
 to send a mail, whereas the solution of the puzzle can be verified efficiently. 
 Later, the concept of verifiable delay functions was
 formalized in~\cite{Dan2018VDF}.

Given the security
parameter $\lambda$ and delay parameter $T$, the prover needs 
to evaluate the $\vdf$ in time $T$. The verifier
verifies the output in $\poly(\lambda,\log T)$-time 
using some proofs produced by the prover. 
A crucial property of $\vdf$s, namely sequentiality,
ensures that the output can not be computed in time much less than $T$ even in the
presence of $\poly(\lambda,T)$-parallelism. 
$\vdf$s have several applications ranging
from non-interactive time-stamping to resource-efficient blockchains, however, are really
rare in practice because of the criteria sequentiality.
In order to design new $\vdf$s we must find problems that offer sequentiality. 
To the best of our knowledge so far, 
all the practical $\vdf$s are based on two inherently sequential algebraic problems 
-- modular exponentiation in groups of unknown
order~\cite{Pietrzak2019Simple,Wesolowski2019Efficient} (fundamentally known as the time-lock puzzle
\cite{Rivest1996Time}) and isogenies over super-singular curves~\cite{Feo2019Isogenie}. 
The security proofs of these $\vdf$s
are essentially polynomial-time reductions from one of these assumptions to the corresponding $\vdf$s.
Thus from the perspective of designers the first hurdle is to find inherently sequential
problems. 
 
The main motivation behind this study has been where should we search for
such inherently sequential problems in order to design new $\vdf$s?
Surprisingly, in this paper, we show that there exist a sequential algorithm to solve
each problem in $\PSPACE$. But we can not derive a $\vdf$ out of each of this problem.
In fact, quite similarly, breaking a strong belief, Mahmoody et al. shows that
the sequence of recursive responses of a random oracle is sequential, 
but no perfectly sound $\vdf$ can be designed using such a sequence only \cite{Mahmoody2020RO}.

\subsection{Proof Sketch}\label{contributions}
We show that the class of all $\vdf$s, $\VDF\subsetneq\PSPACE$ in two parts. In order to proof the inclusion 
$\VDF \subseteq \PSPACE$ first we model $\vdf$s as a special case of interactive proofs,
thus $\VDF \subseteq \IP$. Therefore $\VDF \subseteq \PSPACE$ by the virtue of the
seminal result by Shamir $\IP = \PSPACE$ \cite{ShamirIP}.
We show the opposite exclusion $\PSPACE \nsubseteq \VDF$ by two incorrect attempts to
build $\vdf$s from two different $\PSPACE$-complete languages. The most important
finding is that all problems in $\PSPACE$ do not turn out to be subexponentially
sequential and computationally sound, at the same time.

We consider \textsc{True-Quantified-Boolean-Formula} ($\TQBF$ in short) and \textsc{SPACEHALT} as 
these $\PSPACE$-complete problems. 
The language $\TQBF$ is the set of fully quantified Boolean formula that are true. The
sumcheck protocol for $\TQBF$ is known to be sound against even computationally
unbounded provers. We show that despite its soundness sumcheck protocol fails to achieve
subexponential sequentiality which is necessary for any $\vdf$. 

The language $\HALT$ is the set of all the tuples $(M,x,1^S)$ such that the deterministic 
Turing machine $M$ halts on input $x$ in space $S$. 
It is a $\PSPACE$-complete language as any language $\L \in \PSPACE$ can be reduced
to $\HALT$ in polynomial-time. For any $x\in \L$, the reduction $\L \le_p \HALT$ is nothing but
$f(x) = (M,x,1^{\O(|x|)})$. Moreover, $\HALT$ is an inherently sequential language in a
sense that $\HALT$ can not be parallelized. If it could be then all the languages in 
$\PSPACE$ could be parallelized by the definition of completeness of a language for a class.
But we already know the existence of inherently sequential languages
(e.g., time-lock puzzle) that can be recognized in polynomial space. 

We prove that $\vdf$ derived from $\HALT$ may be subexponentially sequential but not
computationally sound.

\section{Related Work}\label{literature}
In this section, we mention some well-known schemes qualified as $\vdf$s.

The pricing function by Dwork--Naor scheme~\cite{Dwork1992Price} asks a prover,
given a prime $p\equiv3\; (\mathbf{mod}\; 4)$ and a quadratic residue $x$ modulo $p$,
to find a $y$ such that $y^2\equiv x\;(\mathbf{mod}\;p)$. The prover
has no other choice other than using the identity $y\equiv x^{\frac{(p+1)}{4}}\;(\mathbf{mod}\;p)$, 
but the verifier verifies the correctness using $y^2\equiv x\;(\mathbf{mod}\;p)$. 
Evidently, it is difficult to generate difficult instances of this $\vdf$
without using larger primes $p$. Further the massive parallelism with the prover 
violates its sequentiality.

In 2018, Dan et al.~\cite{Dan2018VDF} propose a $\vdf$ based on injective rational
maps of degree $T$, where the fastest possible inversion is to compute the polynomial
GCD of degree-$T$ polynomials. They conjecture that 
it achieves $(T^2,o(T))$ sequentiality using permutation polynomials as the candidate map.
However, it is a weak form of $\vdf$ as the prover needs $\O(T)$-parallelism in order to
evaluate the $\vdf$ in time $T$.

Rivest, Shamir, and Wagner~\cite{Rivest1996Time} introduced another discipline 
of $\vdf$s known as time-lock puzzle.
These puzzles enables an encryption that can be decrypted 
only sequentially. Starting with $N=pq$ such that $p,q$ are large primes,
the key $y$ is enumerated as $y\equiv x^{2^T}\;(\mathbf{mod}\;N)$. Then the verifier,
uses the value of $\phi(N)$ to reduce the exponent to
$e\equiv {2^T}\;(\mathbf{mod}\;\phi(N))$ and finds out $y\equiv x^e\;(\mathbf{mod}\;N)$.
On the contrary, without the knowledge of $\phi(N)$, the only option available to the prover
is to raise $x$ to the power $2^T$ sequentially. 
As the verification stands upon a secret, the knowledge of $\phi(N)$, 
it is not a VDF as verification should depend only on public parameters.

Wesolowski~\cite{Wesolowski2019Efficient} and Pietrzak~\cite{Pietrzak2019Simple} circumvent 
this issue independently. The first one asks the prover to compute an output $y=x^{2^T}$ and
a proof $\pi=x^{\lfloor2^T/l\rfloor}$, where $l$ is a $2\lambda$-bit prime chosen at random. The verifier
checks if $y=\pi^l \cdot x^{(2^T\; \mathbf{mod}\; l)}$. 
Hence the verification needs at most $2\log \ell=4 \lambda$ squaring.
Two candidate groups suits well in 
this scheme -- an RSA group $\multgroup{N}$, and
the class group of an imaginary quadratic number field.
This $\vdf$ shines for its short proof which is a single element in underlying group. 

Pietrzak's $\vdf$ exploits the identity $z^ry=(x^rz)^{2^{T/2}}$
where $y=x^{2^T}$, $z=x^{2^{T/2}}$ and $r \in \{1, \ldots, 2^\lambda\}$ is chosen at random.
So the prover is asked to compute the proof $\pi=\{u_1, u_2, \ldots, u_{\log T}\}$ such that 
$u_i=x^{r_i+{2^{T/2^i}}}$. The verifier computes the $v_i=x^{r_i\cdot{2^{T/2^i}}+2^T}$ and checks if $v_i=u_i^2$.
So the verifier needs $O(\log{T})$ time.
Trading-off the size of the proof it optimizes the effort to generate the proof $\pi$ in $O(\sqrt{T} \log T)$. 
As a comparison, Wesolowski's $\vdf$ needs $\O(T/\log T)$ time to do the same.
This $\vdf$ uses the RSA group and the class groups of imaginary quadratic number fields.

Feo et al.~\cite{Feo2019Isogenie} presents two $\vdf$s based on isogenies of super-singular
elliptic curves. They start with five groups $\langle G_1,G_2,G_3,G_4,G_5\rangle$
of prime order $N$ with two non-degenerate bilinear pairing maps
$e_{12}: G_1 \times G_2 \rightarrow G_5$ and $e_{34}: G_3 \times G_4 \rightarrow G_5$.
Also there are two group isomorphisms
$\phi: G_1 \rightarrow G_3$ and $\overline{\phi}: G_4 \rightarrow G_2$. 
Given all the above descriptions as the public parameters along with a generator $P\in G_1$,
the prover needs to find $\overline{\phi}(Q)$, where $Q\in G_4$, using $T$ sequential steps.
The verifier checks if $e_{12}(P,\overline{\phi}(Q))=e_{34}(\phi(P),Q)$ in $\poly(\log{T})$
time. 
It runs on super-singular curves over $\mathbb{F}_p$ and $\mathbb{F}_{p^2}$ 
as two candidate groups. 
While being inherently non-interactive, a major drawback 
as mentioned by the authors themselves is that its
 setup may be as slow as the evaluation.

Mahmoody et al.~\cite{Mahmoody2020RO}
have recently ruled out the possibility of having perfectly unique $\vdf$s using random oracles only.

 \begin{table}[h]
 \caption{Comparison among the existing VDFs. $T$ is the targeted time bound, $\lambda$ is 
 the security parameter, $\Delta$ is the number of processors. All the quantities may be
 subjected to $\mathcal{O}$-notation, if needed.}
 \label{tab : VDF}
  \centering
  \begin{tabular}{|l@{\quad}|r@{\quad}|r@{\quad}|r@{\quad}|r@{\quad}|r@{\quad}|r@{\quad}}
     \hline
         VDF & \textsf{Eval} & \textsf{Eval} &  \textsf{Verify} & \textsf{Setup} & Proof   \\
 (by authors)&  Sequential   & Parallel      &                  &                &  size   \\
     \hline
     
     Dwork and Naor~\cite{Dwork1992Price}       & $T$   & $T^{2/3}$  &  $T^{2/3}$  & $T$ & $\textendash$ \\ 
     [0.3 em] \hline 
         
     Dan et al.~\cite{Dan2018VDF}         & $T^2$ & $>T-o(T)$  &  $\log{T}$  & $\log{T}$ & $\textendash$  \\
     [0.3 em] \hline
     
     Wesolowski~\cite{Wesolowski2019Efficient} & $(1+\frac{2}{\log{T}})T$   & $(1+\frac{2}{\Delta\log{T}})T$  &  $\lambda^{4}$  & $\lambda^{3}$ & $\lambda^{3}$ \\
     [0.3 em] \hline
     
     Pietrzak~\cite{Pietrzak2019Simple}        & $(1+\frac{2}{\sqrt{T}})T$   & $(1+\frac{2}{\Delta\sqrt{T}})T$  &  $\log{T}$  & $\lambda^{3}$ & $\log{T}$ \\
     [0.3 em] \hline 
     
     Feo et al.~\cite{Feo2019Isogenie}         & $T$   & $T$  &  $\lambda^4$  & $T\log{\lambda}$ & \textendash \\
     [0.3 em] \hline 
     
  \end{tabular}
 \end{table}

\section{Preliminaries}\label{preliminaries}

We start with the notations.

\subsection{Notations}
We denote the security parameter with $\lambda\in\mathbb{Z}^+$.
The term $\poly(\lambda)$ refers to some polynomial of $\lambda$, and
$\negl(\lambda)$ represents some function $\lambda^{-\omega(1)}$.
If any randomized algorithm $\mathcal{A}$ outputs $y$ on an input $x$, 
we write $y\xleftarrow{R}\mathcal{A}(x)$. By $x\xleftarrow{\$}\mathcal{X}$,
we mean that $x$ is sampled uniformly at random from $\mathcal{X}$. For a string $x$, 
$|x|$ denotes the bit-length of $x$, whereas for any set $\mathcal{X}$, $|\mathcal{X}|$ denotes 
the cardinality of the set $\mathcal{X}$. If $x$ is a string then $x[i \ldots j]$ denotes the substring 
starting from the literal $x[i]$ ending at the literal $x[j]$. 
We consider an algorithm $\adv$ as efficient if it runs in 
probabilistic polynomial time (PPT). 
 
\subsection{Verifiable Delay Function}\label{VDF}
We borrow this formalization from~\cite{Dan2018VDF}.

\begin{definition}\normalfont{ \bf (Verifiable Delay Function).}
A VDF $\textsf{V} = (\textsf{Setup}, \textsf{Eval}, \textsf{Verify})$ that implements
a function $\mathcal{X}\rightarrow\mathcal{Y}$ is specified by three algorithms.
\begin{itemize}[label=\textbullet]
\item \textsf{Setup}$(1^\lambda, T) \rightarrow \mathbf{pp}$
 is a randomized algorithm that takes as input a security parameter $\lambda$ 
 and a targeted time bound $T$, and produces the public parameters 
 $\mathbf{pp}$. We require \textsf{Setup} to run in $\poly(\lambda,\log{T})$ time.
 
 \item \textsf{Eval}$(\mathbf{pp}, x) \rightarrow (y, \pi)$ takes an input 
 $x\in\mathcal{X}$, and produces an output $y\in\mathcal{Y}$ and a (possibly empty) 
 proof $\pi$. \textsf{Eval} may use random bits to generate the proof 
 $\pi$. For all $\mathbf{pp}$ generated 
 by $\textsf{Setup}(\lambda, T)$ and all $x\in\mathcal{X}$, the algorithm 
 \textsf{Eval}$(\mathbf{pp}, x)$ must run in time $T$.
 
 \item \textsf{Verify}$(\mathbf{pp}, x, y, \pi) \rightarrow \{0, 1\}$ is a 
 deterministic algorithm that takes an input $x\in\mathcal{X}$, an output $y\in\mathcal{Y}$,
 and a proof $\pi$ (if any), and either accepts ($1)$ or rejects ($0)$. 
 The algorithm must run in $\poly(\lambda,\log{T})$ time.
\end{itemize}
\end{definition}

Before we proceed to the security of $\vdf$s we need the precise model of parallel
adversaries \cite{Dan2018VDF}. 
\begin{definition}\normalfont{(\bf Parallel Adversary)}\label{paradv} 
A parallel adversary $\adv=(\adv_0,\adv_1)$ is a pair of non-uniform 
randomized algorithms $\adv_0$ with total running time $\poly(\lambda,T)$, 
and $\adv_1$ which runs in parallel time $\sigma(T)<T-o(T)$ on at 
most $\poly(\lambda,T)$ number of processors.
\end{definition}
Here, $\adv_0$ is a preprocessing algorithm that precomputes some
$\st$ based only on the public parameters, and $\adv_1$ exploits
this additional knowledge to solve in parallel running time $\sigma$ on 
$\poly(\lambda,T)$ processors.

The three desirable properties of a $\vdf$ are now introduced.

\begin{definition}\normalfont{(\bf Correctness)}\label{def: Correctness} 
A $\vdf$ is correct with some error probability $\varepsilon$,
if for all $\lambda, T$, parameters $\mathbf{pp}$, 
and $x\in\mathcal{X}$, we have
\[
\Pr\left[
\begin{array}{l}
\textsf{Verify}(\mathbf{pp},x,y,\pi)=1
\end{array}
\Biggm| \begin{array}{l}
\mathbf{pp}\leftarrow\textsf{Setup}(1^\lambda,T)\\
x\xleftarrow{\$} \mathcal{X}\\
(y,\pi)\leftarrow\textsf{Eval}(\mathbf{pp},x)
\end{array}
\right]
=1 - \negl(\lambda).
\]
\end{definition}

\begin{definition}\normalfont{\bf(Soundness)}\label{def: Soundness} 
A $\vdf$ is sound if for all non-uniform algorithms $\adv$ 
that run in time $\mbox{poly}(T,\lambda)$,
we have
\[
\Pr\left[
\begin{array}{l}
y\ne\textsf{Eval}(\mathbf{pp},x)\\
\textsf{Verify}(\mathbf{pp},x,y,\pi)=1
\end{array}
\Biggm| \begin{array}{l}
\mathbf{pp}\leftarrow\textsf{Setup}(1^\lambda,T)\\
(x,y,\pi)\leftarrow\mathcal{A}(1^\lambda,T,\mathbf{pp})
\end{array}
\right] \le \negl(\lambda).
\]
\end{definition}

We call the $\vdf$ \emph{perfectly} sound if this probability is $0$.

\begin{definition}\normalfont{\bf (Sequentiality)}\label{def: Sequentiality}
A $\vdf$ is $(\Delta,\sigma)$-sequential if there exists no
pair of randomized algorithms $\adv_0$ with total running time
$\mbox{poly}(T,\lambda)$ and $\adv_1$ which runs
in parallel time $\sigma$ on at most $\Delta$ processors, such that
\[
\Pr\left[
\begin{array}{l}
y=\textsf{Eval}(\mathbf{pp},x)
\end{array}
\Biggm| \begin{array}{l}
\mathbf{pp}\leftarrow\textsf{Setup}(1^\lambda,T)\\
\st\leftarrow\mathcal{A}_0(1^\lambda,T,\mathbf{pp})\\
x\xleftarrow{\$}\mathcal{X}\\
y\leftarrow\mathcal{A}_1(\st,x)
\end{array}
\right]
\le \negl(\lambda).
\]
\end{definition}

We reiterate an important result from~\cite{Dan2018VDF} but as a lemma.
\begin{lemma}{\normalfont ($T \in \mathsf{SUBEXP(\lambda)}).$}\label{subexp}
If $T > 2^{o(\lambda)}$ then there exists an adversary that breaks the sequentiality of
the $\vdf$ with non-negligible advantage.
\end{lemma}
\begin{proof}
$\adv$ observes that the algorithm $\verify$ is efficient.
So given a statement $x \in \X$,  $\adv$ chooses an arbitrary $y \in \Y$ as the output
without running $\eval(x,\pp,T)$. Now, $\adv$ finds the proof $\pi$ by a brute-force
search in the entire solution space with its $\poly(T)$ number of processors.
In each of its processors, $\adv$ checks if $\verify(x,\pp,T,y,\pi_i)=1$ with different
$\pi_i$. The advantage of $\adv$ is $\poly(T)/2^{\Omega(\lambda)} \ge \negl(\lambda)$ as
$T > 2^{o(\lambda)}$. 
\qed
\end{proof}
So we need $T \le 2^{o(\lambda)}$ to restrict the advantage of $\adv$ 
upto $2^{o(\lambda)}/2^{\Omega(\lambda)}=2^{-\Omega(\lambda)}$.

\subsection{The Complexity Classes}

We start with the definition of Turing machine in order to discuss complexity classes.
We consider Turing machines with a read-only input tape and read-write work tape. 
\begin{definition}{\normalfont \textbf{(Turing machine).}}\label{TM}
A Turing Machine is a tuple $\mathsf{TM} = \langle Q, \Gamma, q_0, F, \delta \rangle$ with the following meaning,
\begin{enumerate}
 \item $Q$ is the finite and nonempty set of states.
 \item $\Gamma$ is the finite and non-empty set of tape alphabet symbols including the input alphabet $\Sigma$.
 \item $q_0\in Q$ is the initial state.
 \item $F \subseteq Q$ is the set of halting states.
 \item $\delta : \{Q\setminus F\} \times \Gamma \rightarrow Q \times \Gamma \times D$ is the transition functions 
 where $D=\{-1,0,+1\}$ is the set of directions along the tape.
\end{enumerate}
\end{definition}

Throughout the paper we assume that the initial state $q_0$, one of the final states $q_F$ and the tape alphabet 
$\Gamma=\{0,1, \vdash\}$ are implicit to the description of a $\TM$. Here $\Sigma=\{0,1\}$ and 
$\vdash$ marks the left-end of the tape. Thus $ \langle Q,F,\delta \rangle$ suffices to describe any $\TM$.
\begin{definition}{\normalfont \textbf{(Configuration).}}
 A configuration of a $\TM$ is a triple $(q,z,n)$ where, at present,
 \begin{enumerate}
  \item $q \in Q$ is the state of $\TM$.
  \item $z \in \Gamma^*$ is the content of the tape.
  \item $n \in \Z$ is the position of the head at the tape. 
 \end{enumerate}

\end{definition}

$(q_0,x,0)$ denotes the starting configuration for an input 
string $x$ instead of $(q_0,\vdash x,0)$ (w.l.o.g.).

\begin{definition}{\normalfont \textbf{($\tau$-th Configuration $\stackrel{\tau}{\rightarrow})$.}}
 The relation $\stackrel{\tau}{\rightarrow}$ is defined as,
 \begin{enumerate}
  \item $(q,z,n) \stackrel{\tau}{\rightarrow} (q', z' ,n+d)$ 
  where $z'=\ldots z[n-1]\| b  \|z[n+1] \ldots$ if $\delta(q,z[n])=(q',b,d)$.
  \item $\alpha \stackrel{\tau+1}{\rightarrow} \gamma$ if there exists a $\beta$ such that 
  $\alpha \stackrel{\tau}{\rightarrow} \beta \stackrel{1}{\rightarrow} \gamma$.
 \end{enumerate}
\end{definition}
We denote $\alpha \stackrel{0}{\rightarrow} \alpha$.
 and $\alpha \stackrel{*}{\rightarrow} \beta$ if $\alpha \stackrel{\tau}{\rightarrow} \beta$ for some $\tau \ge 0$.
We say that the $\TM$ halts on a string $x$ with the output $y$ if
$(q_0,x,0)\stackrel{*}{\rightarrow}(q_F,y,n)$ such that $q_F \in F$.

\begin{definition}{\normalfont \textbf{(Time and Space Complexity).}}
We say that a $\TM$ computes a function 
$f:\Sigma^* \rightarrow \Sigma^*$ in time $\tau$ and space $\sigma$
if $\forall x\in \Sigma^*$, $(q_0,x,0)\stackrel{\tau}{\rightarrow}(q_F,f(x),n)$
 using (at most) $\sigma$ different cells
on the working tape (excluding the input tape).
\end{definition}

We call a language $\L$ is decidable by a $\TM$ if and only if there exists a $\TM$ that accepts 
all the strings belong to $\L$ and rejects all the strings belong to $\overline{\L}=\Sigma^*\setminus \L$.
We say that a language $\L$ is reducible to another language $\L'$ if and only if there exists a function $f$ 
such that $f(x) \in \L'$ if and only if $x \in \L$. 
If the function $f$ is computable in $\poly(|x|)$-time then we call it 
as a polynomial time reduction $\L \le _p \L'$.

In order to discuss the complexity classes we follow the definitions provided in~\cite{Arora2009Modern}.

\begin{definition}{\normalfont \textbf{(DSPACE).}}
Suppose $f: \mathbb{N}\rightarrow\mathbb{N}$ be some function. A language $\L$
is in $\DSPACE[f(n)]$ if and only if there is a $\TM$ that decides $\L$ in space $\O(f(n))$.
\end{definition}


\begin{definition}{\normalfont \textbf{(The Class $\PSPACE$).}}
$$\PSPACE=\DSPACE[\poly(n)].$$
\end{definition}

%

\begin{definition}{\normalfont \textbf{($\PSPACE$-complete).}}
A language is $\PSPACE$-complete if it is in $\PSPACE$ and every language in $\PSPACE$
is reducible to it in polynomial time.
\end{definition}

\subsection{Interactive Proof System}
Goldwasser et al. were the first to show that the interactions between the prover and
randomized verifier recognizes class of languages larger than $\NP$~\cite{Goldwasser85Knowledge}. 
They named the class as $\IP$ and the model of
interactions as the interactive proof system. Babai and Moran introduced the same notion of interactions in
the name of Arthur-Merlin games however with a restriction on the verifiers' side \cite{Babai88AM}. 
Later, Goldwasser and Sipser proved that both the models are equivalent
\cite{Goldwasser86IPAM}. Two important works in this context that motivate our present
study are by the Shamir showing that $\IP=\PSPACE$ \cite{ShamirIP} and by the Goldwasser
et al. proving that $\PSPACE = \mathbf{ZK}$, the set of all zero-knowledge protocols.
We summarize the interactive proof system from \cite{Goldwasser88ZK}.

An interactive proof system $\ips$ consists of a pair of $\TM$s, $\prv$ and $\vrf$,
with common alphabet $\Sigma=\{0,1\}$. $\prv$ and $\vrf$ each have distinguished initial and
quiescent states. $\vrf$ has distinguished halting states out of which there is no
transitions. $\prv$ and $\vrf$ operates on various one-way infinite tapes,
\begin{enumerate}[label=\roman*.]
\item $\prv$ and $\vrf$ have a common read-only input tape.
\item $\prv$ and $\vrf$ each have a private random tape and a private work tape.
\item $\prv$ and $\vrf$ have a common communication tape.
\item $\vrf$ is polynomially time-bounded. This means $\vrf$ halts on input $x$ in time
$\poly(|x|)$. $\vrf$ is in quiescent state when $\prv$ is running.
\item $\prv$ is computationally unbounded but runs in finite time. This means $\prv$ 
may compute any arbitrary function $\{0,1\}^*\rightarrow \{0,1\}^*$ on input $x$ 
in time $f(|x|)$. Feldman proved that ``the optimum prover lives in $\PSPACE$"
\footnote{We could not find a valid citation.}.
\item The length of the messages written by $\prv$ into the common communication tape is
bounded by $\poly(|x|)$. Since $\vrf$ runs in $\poly(|x|)$ time, it can not write
messages longer than $\poly(|x|)$.
\end{enumerate}

Execution begins with $\prv$ in its quiescent state and $\vrf$ in its start state.
$\vrf$'s entering its quiescent state arouses $\prv$, causing it to transition to its
start state. Likewise, $\prv$'s entering its quiescent state causes $\vrf$ to
transition to its start state. Execution terminate when $\vrf$ enters in its halting
states. Thus $\ips(x)=1$ denotes $\vrf$ accepts $x$ and $\ips(x)=0$ denotes $\vrf$
rejects $x$.

\begin{definition}{\normalfont \textbf{(Interactive Proof System $\ips$ ).}}
$\ips$ is an interactive proof system for the language $\L\subseteq \{0,1\}^*$ if 

\begin{description}
 \item \noindent {\normalfont (Correctness).} $(x \in \L) \implies \Pr[\ips(x))=1] \ge 1-\negl(|x|)$.
 \item \noindent {\normalfont (Soundness).}  $(x\notin\L) \implies
\forall\prv',\Pr[(\prv'\leftrightarrow\vrf)(x))=1] < \negl(|x|)$.
\end{description}
\end{definition}

The class of interactive polynomial-time $\IP$ is defined as the class of the languages
that have an interactive proof system. Thus $$\IP=\{\L \mid \L \text{ has an }\ips\}.$$
Alternatively and more specifically,

\begin{definition}{\normalfont \textbf{(The Class IP).}}
$$\IP=\IP[\poly(n)].$$
For every $k$, $\IP[k]$ is the set of languages $\L$ such that there exist a probabilistic polynomial time $\TM$
$\vrf$ that can have a $k$-round interaction with a prover $\prv : \{0,1\}^*\rightarrow \{0,1\}^*$ 
having these two following properties
\begin{description}
 \item \noindent {\normalfont (Correctness).} $(x \in \L) \implies \Pr[\ips(x))=1] \ge 1-\negl(|x|)$.
 \item \noindent {\normalfont (Soundness).}  $(x\notin\L) \implies
\forall\prv',\Pr[(\prv'\leftrightarrow\vrf)(x))=1] < \negl(|x|)$.
\end{description}
\end{definition}

%

\section{Fiat--Shamir Transformation}
Any interactive protocol $\L\in\IP$ can be transformed into a non-interactive
protocol if the messages from the verifier $\vrf$ are replaced with the response of
a random oracle $\RO$. This is known as Fiat--Shamir transformation (FS)~\cite{FS86}. In
particular, the $i$-th message from $\vrf$ is computed as
$y_i:=\RO(x,x_1,y_1,\ldots,x_i, y_{i-1})$ where $x_i$ denotes the $i$-th
response of $\prv$. When $\RO$ is specified in the public parameters of a $k$-round
protocol, the transcript $x,x_1,y_1,\ldots,x_k, y_{k-1}$ can be verified publicly.
Thus, relative to a random oracle $\RO$, a $k$-round interactive proof protocol 
$(\prv \leftrightarrow \vrf)$ can be transformed into a two-round non-interactive argument 
$(\prv_{FS} \leftrightarrow \vrf_{FS})$ where $\prv_{FS}$ sends the entire transcript 
$x,x_1,y_1,\ldots,x_k, y_k$ to $\vrf_{FS}$ in a single round. Under the assumption that $\RO$ is
one-way and collision-resistant, $\vrf_{FS}$ accepts $x\in \L$ in the next round if and only if 
$\vrf$ accepts. Here we summarize two claims on Fiat--Shamir transformation stated in~\cite{Ephraim20VDF}.

\begin{lemma}\label{FS}
If there exists an adversary $\adv$ who breaks the soundness of the non-interactive
protocol $(\prv_{FS} \leftrightarrow \vrf_{FS})$ with the probability $p$ using
$q$ queries to a random oracle then there exists another adversary $\adv'$ who breaks
the soundness of the $k$-round interactive protocol $(\prv \leftrightarrow \vrf)$ 
with the probability $p/q^k$.
\end{lemma}

\begin{proof}
See~\cite{Goldreich96} for details. \qed
\end{proof}

\begin{lemma}\label{soundFS}
Against all non-uniform probabilistic polynomial-time adversaries,  
if a $k$-round interactive protocol $(\prv \leftrightarrow \vrf)$
achieves $\negl(|x|^k)$-soundness then the non-interactive protocol 
$(\prv_{FS} \leftrightarrow \vrf_{FS})$ has $\negl(|x|)$-soundness.
\end{lemma}

\begin{proof}
Since, all the adversaries run in probabilistic polynomial time, the number of queries
$q$ to the random oracle must be upper-bounded by $\poly(|x|)$.
Putting $q =|x|^c $ for any $c \in \Z^+$ in lemma.~\ref{FS}, it follows the claim. 
\qed
\end{proof}

\section{$\vdf$ Characterization}
 
In this section, we investigate the possibility to model $\vdf$s as a language in order
to define its hardness. It seems that there are two hurdles,
\begin{description}
\item [Eliminating Fiat--Shamir]
The prover $\prv$ in Def.~\ref{VDF}, generates the proof $\pi:=f(x,y,T,\RO(x,y,T))$ using Fiat--Shamir
transformation where $y:=\eval(x,\pp,T)$. Unless Fiat--Shamir is eliminated from $\vdf$, its
hardness remains relative to the random oracle $\RO$. Sect.~\ref{ivdf} resolves this
issue.

\item  [Modelling Parallel Adversary] How to model the parallel adversary $\adv$ (Def. \ref{paradv}) 
in terms of computational complexity theory? We model $\adv$ as a special variant of
Turing machines described in Def.~\ref{PTM}.
\end{description}

We address the first issue now.

\subsection{Interactive $\vdf$s}\label{ivdf}

We introduce the interactive $\vdf$s in order to eliminate the Fiat--Shamir.
In the interactive version of a $\vdf$, the $\vrf$ replaces the
randomness of Fiat--Shamir heuristic. 
In particular, a non-interactive $\vdf$ with the Fiat--Shamir transcript 
$\langle x,x_1,y_1,\ldots,x_k, y_k \rangle$ can be translated into an equivalent $k$-round 
interactive $\vdf$ allowing $\vrf$ to choose $y_i$s in each round.

%

\begin{definition}\normalfont{ \bf (Interactive Verifiable Delay Function).} 
An interactive verifiable delay function is a tuple $(\setup, \eval,\mathsf{Open} ,\verify)$ 
that implements a function $\X\rightarrow\Y$ as follows,
\begin{itemize}[label=\textbullet]
\item $\setup(1^\lambda, T) \rightarrow \pp$
 is a randomized algorithm that takes as input a security parameter $\lambda$ 
 and a delay parameter $T$, and produces the public parameters 
 $\pp$ in $\poly(\lambda,\log{T})$ time.
 
 \item $\eval(\pp, x) \rightarrow y$ takes an input 
 $x\in\mathcal{X}$, and produces an output $y\in\mathcal{Y}$. 
 For all $\mathbf{pp}$ generated by $\textsf{Setup}(\lambda, T)$ 
 and all $x\in\mathcal{X}$, the algorithm 
 \textsf{Eval}$(\mathbf{pp}, x)$ must run in time $T$.

 \item $\mathsf{Open}(x,y,\pp,T,t)\rightarrow\pi$ 
 takes the challenge $t$ chosen by $\vrf$ and computes a proof  
 $\pi$ (possibly recursively) in $\poly(\lambda)$ rounds of 
 interaction with $\vrf$. In general, for some $k \in \poly(\lambda)$, 
 $\pi=\{\pi_1, \ldots, \pi_k\}$ can be computed as 
 $\pi_{i+1}:=\mathsf{Open}(x_i,y_i,\pp,T,t_i)$ where 
 $x_i$ and $y_i$ depend on $\pi_i$. Observing $(x_i,y_i,\pi_i)$ 
 in the $i$-th round, $\vrf$ chooses the challenge $t_i$ for 
 the $(i+1)$-th round. Hence, $\mathsf{Open}$ runs for $k$-rounds.
  
 \item \textsf{Verify}$(\mathbf{pp}, x, y, \pi) \rightarrow \{0, 1\}$ is a 
 deterministic algorithm that takes an input $x\in\mathcal{X}$, an output $y\in\mathcal{Y}$,
 and the proof vector $\pi$ (if any), and either accepts ($1)$ or rejects ($0)$. 
 The algorithm must run in $\poly(\lambda,\log{T})$ time.
\end{itemize}
\end{definition}

All the three security properties remain same for the interactive $\vdf$. 
Sequentiality is preserved by the fact that
$\mathsf{Open}$ runs after the computation of $y:=\eval(x,\pp,T)$. For soundness, we rely
on lemma.~\ref{soundFS}. The correctness of interactive $\vdf$s implies the correctness
of the non-interactive version as the randomness that determines the proof is not in the
control of $\prv$. Therefore, an honest prover always convinces $\vrf$. 

%
%
%

Although the interactive $\vdf$s do not make much sense as publicly verifiable proofs in
decentralized distributed networks, it allow us to analyze its hardness irrespective of
any random oracle.  

In order to model parallel adversary, we consider a well-known variant of Turing machine that suits the
context of parallelism. We describe the variant namely parallel Turing machine as
briefly as possible from (Sect. 2 in cf.\cite{Worsch1993Parallel})%
 \subsection{Parallel Turing Machine}
 Intuitively, a parallel Turing machine has multiple control units (\textsf{CU})
 (working collaboratively) with a single head associated with each of them
 working on a common read-only input tape \cite{Worsch1993Parallel}. 
 and a common read-write work tape.
 \begin{definition}{\normalfont \textbf{(Parallel Turing Machine).}}\label{PTM}
  a parallel Turing machine is a tuple 
  $\PTM = \langle Q, \Gamma,\Sigma, q_0, F, \delta \rangle$ where
\begin{enumerate}
 \item $Q$ is the finite and nonempty set of states.
 \item $\Gamma$ is the finite and non-empty set of tape alphabet symbols including the input alphabet $\Sigma$.
 \item $q_0\in Q$ is the initial state.
 \item $F \subseteq Q$ is the set of halting states.
 \item $\delta : 2^Q \times \Gamma\rightarrow 2^{Q\times D} \times \Gamma$ 
 where $D=\{-1,0,+1\}$ is the set of directions along the tape.
\end{enumerate}
 \end{definition}
 
 A configuration of a $\PTM$ is a pair $c=(p,b)$ of mappings
 $p:\Z^+\rightarrow 2^Q$ and $b:\Z^+\rightarrow \Gamma$.
 The mapping $p(i)$ denotes the set of states of the \textsf{CU}s
 currently pointing to the $i$-th cell in the input tape and $b(i)$ is the 
 symbol written on it. So it is impossible for two different \textsf{CU}s 
 pointing to the same cell $i$ while staying at the same state simultaneously.
 During transitions $c'=(M'_i,b'(i))=\delta(c)=\delta(p(i),b(i))$, the set of \textsf{CU}s
 may be replaced by a new set of \textsf{CU}s $M'_i \subseteq Q \times D$. 
 The $p'(i)$ in the configuration $c'$ is defined as 
 $p'(i)=\{q \mid (q,+1)\in M'_{i-1} \lor (q,0)\in M'_{i} \lor (q,-1)\in M'_{i+1}\}$. 
 
 Without loss of generality, the cell $1$ is observed in order to find 
 the halting condition of $\PTM$. We say that a $\PTM$ halts on a 
 string if and only if $p(1)\subseteq F$ after some finite time. 
 The notion of decidability by a $\PTM$ 
 is exactly same as in $\TM$. We denote $\PTM(s,t,h)$
 as the family of all languages for which there is a $\PTM$ recognizing them 
 using space $s$, time $t$ and $h$ processors. Thus languages decidable by a
 $\TM$ is basically decidable by a $\PTM(s,t,1)$. Assuming $\TM(s,t)$ is the set of
 languages recognized by a $\TM$ in space $s$ and time $t$, we mention Theorem 15 from
 (cf. \cite{Worsch1993Parallel}) without the proof.


We observe that the parallel adversary $\adv$ defined in Def. \ref{paradv} is
essentially a $\PTM$ having $\poly(\lambda,T)$ processors running on $\poly(\lambda,T)$
space in time $\sigma(T)$. We will refer such a $\PTM$ with $\poly(\lambda,T)$-$\PTM$
(w.l.o.g.) in our subsequent discussions.

\subsection{$\vdf$ As A Language} 

Now  we characterize $\vdf$s in terms of computational complexity theory. We
observe that, much like $\ips$, $\vdf$s are also proof system for the languages,
\[
\L=\left\{(x,y,T)
\begin{array}{l}
\end{array}
\Biggm| \begin{array}{l}
\mathbf{pp}\leftarrow\setup(1^\lambda,T)\\
x \in  \{0,1\}^\lambda\\
y\leftarrow\eval(\pp,x)
\end{array}
\right\}.
\]
$\prv$ tries to convince $\vrf$ that the tuple $(x,y,T)\in \L$ in
polynomially many rounds of interactions. 
In fact, Pietrzak represents his $\vdf$ using such a language (Sect. 4.2 in cf.
\cite{Pietrzak2019Simple}) where it needs $\log T$ (i.e., $\poly(\lambda))$
rounds of interaction. However, by design, the $\vdf$ is non-interactive. It uses
Fiat--Shamir transformation.

Thus, a $\vdf$ closely resembles an $\ips$ except on the fact
that it stands sequential (see Def.~\ref{def: Sequentiality}) even against 
an adversary (including $\prv$) possessing subexponential parallelism. 
Notice that a $\poly(\lambda,T)$-$\PTM$ (see Def. \ref{PTM})
precisely models the parallel adversary described in Def. \ref{paradv}. 
In case of interactive proof systems, we never talk about the running time of $\prv$
except its finiteness. On the contrary, $\prv$ of a $\vdf$ must run for at least
$T$ time in order to satisfy its sequentiality. Hence, we define $\vdf$ as follows,

\begin{definition}{\normalfont \textbf{(Verifiable Delay Function $\vdfs$).}}
For every $\lambda\in \Z^+$, $T \in 2^{o(\lambda)}$ and for all $s= (x,y,T) \in
\{0,1\}^{2\lambda +\lceil\log T\rceil}$,
$\vdfs$ is a verifiable delay function for a language $\L\subseteq\{0,1\}^*$ if
\begin{description}
 \item \noindent {\normalfont (Correctness).} $(s \in \L) \implies \Pr[\vdfs(s))=1]\ge 1-\negl(\lambda)$.
 \item \noindent {\normalfont (Soundness).} $(s \notin\L)\implies\forall\adv,\Pr[\advs(s))=1] \le \negl(\lambda)$.
 \item \noindent {\normalfont (Sequentiality).} $(s\in\L)\implies\forall\adf,\Pr[\adfs(s))=1] \le \negl(\lambda)$.
\end{description}
where, 
\begin{enumerate}[label=\roman*.]
\item $\prv : \{0,1\}^*\rightarrow \{0,1\}^*$ is a $\TM$ that runs in time $\ge T$, 
\item $\adv : \{0,1\}^* \rightarrow \{0,1\}^*$ is a $\TM$ that runs in 
time $\poly(\lambda,T)$,
\item $\adf$ is a {\normalfont $\poly(\lambda,T)$}-$\PTM$ 
(see Def. \ref{PTM}) that runs in time $<T$.
\end{enumerate}
\end{definition}


Further we define the class of all verifiable delay functions as, 
\begin{definition}{\normalfont \textbf{(The Class $\VDF)$.}}\label{VDF}
$$\VDF=\VDF[\poly(\lambda)].$$
For every $k\in\Z^+$, $\VDF[k]$ is the set of languages $\L$ 
such that there exists a probabilistic polynomial-time 
$\TM$ $\vrf$ that can have a {\normalfont $k$}-round interaction with
\begin{enumerate}[label=\roman*.]
\item $\prv : \{0,1\}^* \rightarrow \{0,1\}^*$ is a $\TM$ that runs in time $\ge T$, 
\item $\adv : \{0,1\}^* \rightarrow \{0,1\}^*$ is a $\TM$ that runs in time $\poly(\lambda,T)$,
\item $\adf$ is a {\normalfont $\poly(\lambda,T)$}-$\PTM$ 
(see Def. \ref{PTM}) that runs in time $<T$.
\end{enumerate}
satisfying these three following properties,

\begin{description}
 \item \noindent {\normalfont (Correctness).} $(s \in \L) \implies \Pr[\vdfs(s))=1]\ge 1-\negl(\lambda)$.
 \item \noindent {\normalfont (Soundness).} $(s \notin\L)\implies\forall\adv,\Pr[\advs(s))=1] \le \negl(\lambda)$.
 \item \noindent {\normalfont (Sequentiality).} $(s\in\L)\implies\forall\adf,\Pr[\adfs(s))=1] \le \negl(\lambda)$.
\end{description}
\end{definition}

\section{$\PSPACE$-hardness}

Although the definitions of $\IP$ and $\VDF$ appear quite similar, 
the key differences are,
\begin{enumerate}
\item $\IP$ does not demand for sequentiality. In fact, one of the most elegant $\IP$,
the sumcheck protocol for $\mathsf{UNSAT}$ is known to be parallelizable in nature. The
sumcheck protocol asks $\prv$ to compute the sum $\sum_{z\in \{0,1\}^n}f(z)=y$ of
polynomial $f$ of small degree. Therefore, $y$ can be computed in $2^{\O(n)}$-time
sequentially but in $\O(2^n/\Gamma)$-time parallelly when $\prv'$ has $\Gamma$ number of
parallel processors. As $\IP$ allows a malicious prover $\prv'$ to have unbounded
computational power (so processors), sequentiality can not be achieved against $\prv'$. 

\item $\IP$ demands for statistical soundness i.e., no prover
has non-negligible advantage to convince $\vrf$ with a false proof. 
On the contrary, $\VDF$ asks for computational soundness only
i.e., no prover running in $\poly(\lambda,T)$-time has non-negligible
advantage. As we see, statistical soundness implies computational soundness. 
\end{enumerate}
 
These two observations together imply that $\vdf$s are special kind of interactive
proofs that are sequential but with a bit relaxed notion of soundness.  

Therefore, the proof for $\VDF \subseteq \IP = \PSPACE$ is straightforward. 

\begin{theorem}
 $\VDF \subseteq  \PSPACE$.
\end{theorem}
\begin{proof}
An honest prover $\prv$ needs to run for time $T$ 
to decide if a tuple $(x,y,T)\in\L$ for all $\L\in \VDF$.
By lemma.~\ref{subexp}, $T$ can be at most $2^{o(\lambda)}$. 
By lemma.~\ref{halttime}, a $\TM$ with $S$-space may run for $|Q|S2^S$-time.
Therefore, a $\TM$ with even $o(\lambda)$-space suffices to decide $\L$. 
Hence, $\VDF \subseteq \PSPACE$.
\qed
\end{proof}

The standard way to prove $\PSPACE \subseteq \VDF$ is to derive a $\vdf$
from a $\PSPACE$-complete language~\cite{ShamirIP}. Existence 
of such a $\vdf$ would imply that there is an inherently sequential 
$\PSPACE$-complete problem whose solution is sound also.
We claim that such a $\PSPACE$-complete problem that commits 
sequentiality and soundness together, hardly exists. We present two flawed $\vdf$s
in order to show this.

\subsection{A Sound But Non-sequential Approach}
Our goal is to check if a statistically sound interactive proof evokes a $\vdf$.
Thus, following~\cite{ShamirIP}, we attempt to derive a $\vdf$ from $\TQBF$. 

\begin{definition}{\normalfont (\textsc{True-Quantified-Boolean-Formula} $\TQBF$).}
Let, $\Psi=Q_1x_1,\ldots Q_nx_n \phi(x_1, \ldots, x_n)$ be a quantified Boolean formula
of $n$ variables and $m$ clauses such that all $Q_i \in\{\exists, \forall\}$ and $\phi$ is in 3-CNF
(w.l.o.g.). The language $\TQBF$ is defined as the set of all quantified Boolean 
formula that are true. Formally,

$$ \TQBF = \{\Psi(x_1, \ldots, x_n)=1 \}.$$
\end{definition}

\subsubsection{Sumcheck Protocol for $\TQBF$}

We summarize the sumcheck protocol for $\TQBF$ from~\cite{ShamirIP}.

Given a quantified Boolean formula (QBF) $\Psi$, first we arithmetize $\Psi$ to obtain
a polynomial $f$ as follows, 
\begin{enumerate}
\item $\forall x_n \phi(x_1,\ldots, x_n)$ evaluate $\prod_{x_n\in \{0,1\}} f(x_1,\ldots, x_n)$.
\item $\exists x_n \phi(x_1,\ldots, x_n)$ evaluate $\sum_{x_n\in \{0,1\}} f(x_1,\ldots, x_n)$. 
\end{enumerate}
Thus, a QBF $\Psi=\forall x_1 \exists x_2 \ldots \forall x_n \phi(x_1,\ldots, x_n) \in \TQBF$ 
if and only if 
$h(x_1,\ldots, x_n)=\prod_{x_1\in \{0,1\}} \sum_{x_2\in \{0,1\}} \ldots \prod_{x_n\in \{0,1\}} f(x_1,\ldots, x_n) \ne 0$. 
As $\phi$ is in 3-CNF, degree of $f$ 
is $\O(n^3)$.But, due to the presence of $\prod$ operator in $h$, its degree and number
of coefficients can be $\O(2^n)$ in the worst-case.

It is resolved with the observation that $x^k=x$ for all $k \ge 1$ as $x\in\{0,1\}$. It
allows to define a linearization operator $L_n$ as follows,
$$ L_n f(x_1,\ldots, x_n)=x_n\cdot f(x_1,\ldots, x_{n-1},1)+(1-x_n)\cdot f(x_1,\ldots, x_{n-1},0).$$ 
Thus, in order to keep the degree and size of $h$ in $\poly(n)$, we sprinkle the
linearization operators in between $h$ as,
$$h'=\prod_{x_1\in \{0,1\}}L_1 \sum_{x_2\in \{0,1\}} L_1 L_2
\prod_{x_3\in \{0,1\}} \ldots \prod_{x_n\in \{0,1\}}L_1 L_2 \ldots L_n  f(x_1,\ldots, x_n).$$
The size of $h'$ is $\O(n^2)$ as there are exactly $n(n+3)/2$ operators. 

The sumcheck protocol for $\TQBF$ asks the prover $\prv$ to prove that 
$h'(x_1,\ldots, x_n)=y \ne 0 \mod{p}$ for a prime $p \ge 2^n 3^m$ in at most
$n(n+3)/2$-rounds. In each round, $\prv$ strips one operator.
Let us denote the operator before $x_i$ in $h'$ with $\otimes_i \in
\{\prod_i,\sum_i, L_i\}$. The protocol is defined recursively as follows,

Suppose the partial sum of $h'(r_1, r_2, \ldots, r_{i-1}, x_i,\ldots x_n)=y'$
where $r_j$ is chosen from the finite field $\F_p$ uniformly at random. For all $i$,
\begin{description}
\item [Case 1:] If $\otimes_i = \sum_{x_i}$ then $\prv$ sends a univariate polynomial
$s(x_i)=h'(r_1, r_2, \ldots, r_{i-1}, x_i,\ldots, x_n)$. The verifier $\vrf$ rejects if
$s(0)+s(1) \ne y'$, otherwise asks $\prv$ to prove in the next round that 
$s(r_i)=h'(r_1, r_2, \ldots, r_{i-1}, r_i,x_{i+1},\ldots, x_n)$ for some $r_i \in_R \F_p$.

\item [Case 2:] If $\otimes_i = \prod_{x_i}$ Exactly same as case (1) except that $\vrf$ 
rejects if $s(0)\cdot s(1) \ne y'$ instead of $s(0)+s(1)\ne y'$.

\item [Case 3:] If $\otimes_i = L_{x_i}$ then $\prv$ sends a univariate polynomial
$s(x_i)=h'(r_1, r_2, \ldots, r_{i-1}, x_i,\ldots, x_n)$. The verifier $\vrf$ rejects if
$r_i\cdot s(0)+(1-r_i)\cdot s(1) \ne y'$, otherwise asks $\prv$ to prove in the next round 
that $s(r_i)=h'(r_1, r_2, \ldots, r_{i-1}, r_i,x_{i+1},\ldots, x_n)$ for some 
$r_i \in_R \F_p$.
\end{description}

\begin{lemma}
 For all adversaries $\adv$ even with computationally unbounded power the probability the $\vrf$
 accepts a false $y \ne h'(x_1,\ldots, x_n)$ is at most $\frac{3mn+n^2}{p}$.
\end{lemma}
 
\begin{proof}
For case (1) and (2), the degree of the polynomial $s(x_i)$ is $1$ as $L_{x_i}$s
linearize $s$. For case (3), the degree of $s(x_i)$ can be at most $2$.

By Schwartz-Zippel lemma, every two distinct univariate polynomials of degree 
$\le d$ over a field $\F$ agree in at most $d$ points. So, the probability that $\vrf$
accepts a wrong $y$ has this two components,
\begin{enumerate}[label=\roman*.]
\item For the inner $L_{x_i}$, this probability is $\le \frac{2}{p}$.  
\item For the final $L_{x_i}$, this probability is $\le \frac{3m}{p}$.  
\end{enumerate}

Therefore, by the union bound, the total probability that $\vrf$ accepts a wrong $y$
is at most,

$$ \frac{n}{p} + \frac{3mn}{p} + \frac{2}{p}\sum_{i=1}^{n-1} i= \frac{3mn+n^2}{p}.$$ 
\qed
\end{proof}

\subsubsection{Argument Against Sequentiality}
Although $\TQBF$ raises an interactive proof with statistical soundness, we are not sure
if it admits a subexpoentially long sequential computation too. The reason is that the
maximum number of sequential steps required to evaluate $h'$ in the worst-case is the
number of operators in it i.e., $n(n+3)/2$. So, in order to setup a $\vdf$ for
sequential time $T$, it needs to sample a QBF of length $\Omega(\sqrt{T})$ as the public
parameter. To support efficient execution in practice, cryptographic protocols should allow
polynomially long public parameters only. Therefore, such a $\vdf$ may work for 
$T \in\poly(\lambda)$ but not for $T \in 2^{o(\lambda)}$.

\subsection{A Sequential But Unsound Approach}
In this section, we present another wrong attempt to derive a $\vdf$ which turns out to
be sequential but not computationally sound.

\begin{definition}{\normalfont \textbf{($\HALT)$.}}
Suppose $M$ is a Turing machine, $x\in\Sigma^*$ is an input string and $S\in\mathbb{N}$.  
The language $\HALT$\footnote{See $\mathsf{SPACE\; TMSAT}$ (cf. Def. 4.9)
in \cite{Arora2009Modern}.} is the
set of all the tuples $(\langle M \rangle, x, 1^S)$  
such that the $\TM$ $M$ halts on input $x$ in space $S$. Formally, 
$$\HALT = \{ (\langle Q,F,\delta \rangle ,x,1^S) \mid
(q_0,x,0)\xrightarrow{*}(q_F,y,n)\text{ in space } S\}.$$
\end{definition}

\begin{lemma}\label{haltcomp}
 $\HALT$ is $\PSPACE$-complete.
\end{lemma}
\begin{proof}
We show that any language $\L \in \PSPACE$ is reducible to $\HALT$ in polynomial-time.
Suppose $\L \in \PSPACE = \DSPACE (\poly(n))$ is decided by a $\TM$ $M$.
Then the function $f(x)=(\langle M  \rangle, x, 1^{\poly(\mid x \mid )})$
is a polynomial-time reduction from $\L$ to $\HALT$. \qed
\end{proof}

\begin{lemma}\label{haltseq}
If time-lock puzzle is inherently sequential then $\HALT$ is inherently sequential.
\end{lemma}
\begin{proof}
We prove this by contradiction.
Suppose $(\langle M \rangle, x, T)\in \HALT$ and we parallelize the simulation of $M$ 
with another $\TM$ $\widehat{M}$. Then any $\TM$ that decides a language $\L \in \PSPACE$ must be parallelizable
using the $\TM$ $\widehat{M}$. Then it means that there exists no inherently sequential
language in $\PSPACE$. But we already know the existence of languages
(e.g., time-lock puzzle \cite{Rivest1996Time}) which is
sequential but can be evaluated in polynomial space.
 \qed
\end{proof}

\begin{lemma}\label{halttime}
 $\HALT$ is decidable in at most $|Q|S2^S$ time.
\end{lemma}
\begin{proof}
There are $2^S$ different strings that could appear in the work tape of $M$.
The head could be in any of $S$ different places and the $M$ could be
in one of $|Q|$ different states. So the total number of configurations is
$|Q|S2^S$. 
\qed
\end{proof}
By the pigeonhole principle, if $M$ is run for further steps, it must visit
a configuration again resulting into looping.
We design a $\vdf$ from the language $\HALT$.

\subsubsection{$\vdf$ From $\HALT$}\label{haltvdf}
The design of this $\vdf$ is based on two 
fundamental observations on Turing machines. 
\begin{itemize}
\item the running time $T$ of $M$ having $S$ space is bounded by $|Q|S2^S$ 
(by lemma.~\ref{halttime}).
\item $M$ continues to stay within the set of halting states 
(either accepting or rejecting) once it reaches at one of them.
\end{itemize}

We specify the algorithms for $\vdf$s as follows,
\begin{description}
\item [$\setup(1^\lambda, T)\rightarrow \pp$]
 It samples a Turing machine, $M=(Q,F,\delta)$ such that
 $2^{\Omega(\lambda)} \le  |Q| \le 2^{\poly(\lambda)}$ and
 $|F| \le \poly(\lambda)$. Although the description of $M$ 
is exponentially large, it suffices to provide $\delta$ as a $\poly(\lambda)$-size
circuit that outputs the next state $q_{i+1}$ on an input $q_i$. Without loss of
generality, we assume an encoding $Q=\{0,1\}^{\lceil \log |Q| \rceil}$ with 
the implicit initial state $q_0=0^{\lceil \log |Q| \rceil}$.

\item[$\eval(x, \pp)\rightarrow (y,\pi)$]
The prover $\prv$ computes the $T$-th state $q_T$ of $M$ on the input $x$ starting from
the initial state $q_0$ and sends it to the verifier $\vrf$. Formally,
$q_T:=\delta^T(q_0,x)$. Now, $\vrf$ asks $\prv$ to provide another state $q_t$ from this
sequence of states $q_0, \delta(q_0), \ldots, \delta^T(q_0)$ such that $(T-t) \le
\lambda$. With $q_t$, $\vrf$ also asks for the tape content $z$ of $M$ at time $t$.

\begin{algorithm}
\caption{$\eval$ from $\HALT$}\label{alg:eval}
\begin{algorithmic}[1]
\STATE  $(q_0,x,0)\stackrel{T}{\rightarrow}(q_T,y,n)$. 
\STATE Obtains $t \ge T -\lambda$ from $\vrf$. 
\STATE $(q_0,x,0)\stackrel{t}{\rightarrow}(q_{t},y_{t},n_{t})$.
\STATE Initialize a string $z:=y_{t}[n_{t}]$.
\FOR  {$t \le  i \le T$}
\STATE  $(q_{i},y_{i},n_{i})\stackrel{1}{\rightarrow}(q_{i+1},y_{i+1},n_{i+1})$.
\STATE  $z:=z \| y_{i+1}[n_{i+1}]$.
\ENDFOR
\STATE $y:=q_T$
\STATE $\pi:(q_t,z)$
\RETURN $(x,y,T,\pi)$
\end{algorithmic}
\end{algorithm}

\item[$\verify(x, \pp, \phi, \pi)\rightarrow \{0,1\}$]
Using $\verify$, $\vrf$ checks if $y=q_T$ in $(T-t) \le \lambda$ steps, as follows,

\begin{algorithm}
\caption{$\verify$ for $\HALT$}\label{alg:verify}
\begin{algorithmic}[1]
 \STATE $(q_{t},z,0)\stackrel{T-t}{\rightarrow}(q_T,y,n)$.
 \IF {$y=q_T$}
	\RETURN 1;
 \ELSE 
	\RETURN 0; 
 \ENDIF
\end{algorithmic}
\end{algorithm}
\end{description}

\subsection{Efficiency}
Here we discuss the time and the memory required by the prover and the verifier.
\begin{description}
 \item [Proof Size] The output $\phi=q_T$ needs $\log |Q|$-bits where 
 $\lambda \le \log |Q| \le \poly(\lambda)$. The proof $\pi=(q_{T-t},z)$ 
 requires $\log |Q|+t$-bits as $|z|=t$.
 
 \item [$\prv$'s Effort] $\prv$ needs $T$ time to find $q_T$, then $T-t$ time to find $q_{T-t}$ and 
 finally $t$ time to find $z$. By the Theorem~\ref{thm : seq} it requires $2T$ time in total.
 
 \item [$\vrf$'s Effort] $\vrf$ needs only $t$ time to find $q_T$ from $q_{T-t}$ using $z$. 
Deciding $q \in F$ is already shown to be efficient.
\end{description}

\section{Security}
We claim that the derived $\vdf$ is correct and sequential but not sequential.

\begin{theorem}
 The constructed $\vdf$ is correct.
\end{theorem}

\begin{proof}
 If $\prv$ has run $\eval$ honestly then $(q_0,x,0)\stackrel{T}{\rightarrow}(q_T,y,n)$.
 The integer $t$ is solely determined by the input statement $x$ and the output $q_T$. 
 Thus $\vrf$ will always find $y=q_T$ by computing $(q_{t},z,0)\stackrel{T-t}{\rightarrow}(q_T,y,n)$. \qed
\end{proof}

\begin{theorem}\label{thm : seq}
 If there is an adversary $\adv$ breaking the sequentiality of this $\vdf$ with the probability $p$ then 
 there is a Turing machine $\adv'$ breaking the sequentiality of the language $\HALT$ with the same probability.
\end{theorem}

\begin{proof}
 Suppose a $\TM$ halts on input $x$ in space $S$ and in time $T \le |Q|S2^S$.
 Thus to decide if the string $ ( M,x,1^S ) \in \HALT$ in time $<T$,
 $\adv'$ passes it to $\adv$ that runs in time $<T$. 
 When $\adv$ outputs $(q_T,\pi)$, $\adv'$ checks if
 $q_T \in F$ or not and decides the string $(M,x,1^S)$.
 Observe that deciding $q_T \in F$ is efficient as $|F| \in \poly(\lambda)$.
 
 As $\adv$ breaks the sequentiality of this $\vdf$ with the probability $p$ so $\Pr[q_T \in F]=p$. 
 Hence $\adv'$ decides $\HALT$ in time $<T$ with the probability $p$ violating the sequentiality of $\HALT$.\qed
\end{proof}

While everything seems perfect in the above construction of the $\vdf$, there exists an
adversary that breaks the soundness of this $\vdf$ \emph{certainly}. 

\begin{theorem}\label{thm : smart}
There exists an adversary $\adv$ who breaks the soundness of the $\vdf$ derived from
$\HALT$.
\end{theorem}

\begin{proof}
The key observation made by $\adv$ is that the verifier $\vrf$ never checks if the state
$q_{t}$ in proof $\pi$ has correctly been computed starting from the state $q_0$ in
$t$-steps. However, $\adv$ knows that the maximum distance from the state $q_{t}$ and
$q_T$ is upper-bounded by $\lambda$.

Therefore, $\adv$ chooses an arbitrary configuration $\alpha$ of $M$. (S)he never runs
$\eval$. Rather, by simulating $M$, (s)he computes another configuration $\alpha \xrightarrow{\lambda}
\beta$ remembering the sequence of all the states and the scanned input symbols in these $\lambda$ steps. 
This can be done efficiently as $Q$ can be encoded using only $\poly(\lambda)$-bits. 
Suppose, the sequence of states and the scanned symbols are denoted as $\mathcal{Q}$ and $\mathcal{Z}$.
Observe that $|\mathcal{Q}|=|\mathcal{Z}|=\lambda$. 
Finally, $\adv$ announces the state in the configuration $\beta$ as the
output $q_T$

When $\adv$ obtains the offset $t$ from $\vrf$, 
$\adv$ fixes the $q_{t}:=\mathcal{Q}[t-\lambda-1]$ and
$z:=\mathcal{Z}[t-\lambda-1, \ldots, \lambda]$. Essentially, $\adv$ sets the
$(t-\lambda)$-th state in $\mathcal{Q}$ as the state $q_{t}$ in the proof $\pi$. 
Similarly, the tape content $z$ in $\pi$ is the sequence of last $(T-t)$ symbols in 
$\mathcal{Z}$. 

Clearly, $\vrf$ will be convinced as $(q_{t},z,0)\xrightarrow{T-t}(q_T,z',n)$.
 
\qed
\end{proof}
%
%
%

\subsubsection{Argument Against Soundness} Theorem.~\ref{thm : smart} suggests that the
$\vdf$ is not computationally sound, however, is sequential for any $T \in 2^{o(\lambda)}$.
This argument can be supported with an informal claim that if such a sequence of states
can be verified efficiently with soundness then a chain of hashes (i.e., $\RO^T(x)$) may
also raise $\vdf$s. Unfortunately, we do know no efficient verifiaction algorithm
for hash-chains.

\section{Frequently Asked Questions}
In this section, we would like to clear some obvious confusions in our results. 

\begin{itemize}

\item [Q.1)] Does Lemma.~\ref{haltseq} suggest that every problem in $\PSPACE$
is inherently sequential?

\item [A.1)] Yes and no both. Because our interpretation of this theorem is that we call
a problem inherently sequential as long as we do not know an efficiently parallelizable
for the same. For example, the time-lock puzzle or the isogenies over the super-singular 
curves. Similarly, having a parallel algorithm for a problem does not really deny the
chance of having another sequential algorithm for the same. Because a parallelizable
language in $\PSPACE$ also reduces to $\HALT$. So, we believe that every language in
$\PSPACE$ has sequential algorithm and parallelizable algorithm together.

\item [Q.2)] Does the answer A.1 suggest that problems in $\P$ can also be used to
derive $\vdf$?

\item [A.2)] Yes. Because, the key observation is that the definition of $\vdf$ imposes an upper bound on
$T \le 2^{o(\lambda)}$ but not the lower bound. For arbitrarily small $T$ , the computation can not
be arbitrarily hard. For all $T \le \poly(\lambda)$, the hardness can not be beyond $\P$.  In
principle, we need sequentially hard problems rather than the computationally hard
problems to derive VDFs.

For example, circuit value problem (CVP) is known to be a P-complete problem. 
The algorithm that solves CVP is efficient but is sequential also as we do not know faster parallel algorithm. 
So, $\vdf$s derived from $\P$-complete problems should work for all $T \in \poly(\lambda)$.

\item [Q.3)] The $\vdf$ derived in Sect.~\ref{haltvdf} uses a random oracle $\RO$ in
$\eval$ and $\verify$. Does it mean that all the subsequent claims are relative to
a random oracle?

\item [A.3)]
The definitions of the class $\VDF$ [Def. 19] and the interactive $\vdf$s [Def. 18] are
independent of any random oracle. As the class VDF is defined to be the set of
interactive VDFs, this result is not relative to any random oracle. It could have been
only if the class $\VDF$ was defined to be the set of non-interactive VDFs.

The random oracle $\RO$ is required only in the non-interactive version of this $\vdf$ 
as per the Fiat–Shamir heuristic, but not in the interactive VDF. We have explicitly
mentioned how the verifier replaces the random oracle $\RO$ in the interactive version  
in the descriptions of each algorithm $\setup$, $\eval$ and $\verify$.
Still the interactive version satisfies all the security proofs in Sect.~\ref{ivdf} which are
also independent of any random oracle.

\item [Q.5)] Do the flawed attempts to derive $\vdf$s suggest 
that no problems in $\PSPACE$ can be used to derive $\vdf$s?

\item [A.5)] No. They suggest that all problems in $\PSPACE$ do not have
subexpoentially long sequentiality and computational soundness together.
Time-lock puzzle probably one among such rare problems that has both the properties. 
In particular, Pietrzak's $\vdf$~\cite{Pietrzak2019Simple} is known 
to be statistically sound and, of course, sequential. It is derived from time-lock
puzzle. Hence, these flawed attempts actually introduce the notion a subclass $\VDF$ 
full of such special problems within $\PSPACE$.

\end{itemize}
\bibliographystyle{splncs04}
\bibliography{ref}

\end{document}